\documentclass[a4paper]{jpconf}
\usepackage{graphicx}
\usepackage{amsmath}
\usepackage{latexsym}
\usepackage{amsfonts}
\usepackage{amssymb}
\usepackage[dvipsnames]{xcolor}

\usepackage{lmodern,bm}
\usepackage{bbm,dsfont}

\usepackage{graphicx}
\usepackage{amsthm}
\usepackage{verbatim}

\newtheorem{proposition}{Proposition}
\newtheorem{theorem}{Theorem}
\newtheorem{lemma}{Lemma}
\newtheorem{corollary}{Corollary}

\newtheorem{definition}{Definition}

\usepackage{cite}

\usepackage{color}

\usepackage[normalem]{ulem} 

\newcommand{\bhs}{\mathbf{B}(\mathcal{H}_{\mathcal{S}})}
\newcommand{\bhr}{\mathbf{B}(\mathcal{H}_{\mathcal{R}})}
\newcommand{\bht}{\mathbf{B}(\mathcal{H}_{\mathcal{S}}\otimes \mathcal{H}_{\mathcal{R}})}

\newcommand{\hi}{\mathcal{H}} 
\newcommand{\his}{\mathcal{H}_{\mathcal{S}}}
\newcommand{\hir}{\mathcal{H}_{\mathcal{R}}}

\newcommand{\hik}{\mathcal{K}} 
\newcommand{\hr}{\mathcal{H}_{\mathcal{R}}}

\newcommand{\hs}{\mathcal{H}_{\mathcal{S}}}

\newcommand{\id}{\mathbbm{1}} 

\newcommand{\Sy}{\mathcal{S}}

\newcommand{\R}{\mathcal{R}}
\newcommand*\colvec[3][]{
    \begin{pmatrix}\ifx\relax#1\relax\else#1\\\fi#2\\#3\end{pmatrix}
}

\newcommand{\bhh}{\mathbf{B}(\mathcal{H})}

\begin{document}

\title{A quantum reference frame size-accuracy trade-off for quantum channels}

\author{Takayuki Miyadera$^1$ and Leon Loveridge$^2$}

\address{$^1$ Department of Nuclear Engineering, Kyoto University, Nishikyo-ku, Kyoto 615-8540, Japan}
\address{$^2$ Quantum Technology Group, Department of Science and Industry Systems, University of South-Eastern Norway, 3616 Kongsberg, Norway}

\ead{miyadera@nucleng.kyoto-u.ac.jp, leon.d.loveridge@usn.no}

\begin{abstract}
The imposition of symmetry upon the nature and structure of quantum observables has recently
been extensively studied, with quantum reference frames playing a crucial role. In this paper,
we extend this work to quantum transformations, giving quantitative results showing, in direct analogy to the case of observables, that a ``large" reference frame is required for non-covariant channels to be well approximated by covariant ones. We apply our findings to the concrete setting of $SU(2)$ symmetry.
\end{abstract}
\begin{center}
{\bf Dedicated to the memory of our friend Paul Busch}
\end{center}

\section{Introduction} 

The role played by symmetry in the understanding and development of physics can hardly be overstated. It is an important part of the process of building mathematical
models of the physical world, and is often crucial for their solubility 
 in both classical and quantum settings. 
Less widely recognised, however, is the effect that symmetry has in limiting
what is measurable. 
In fact, 
as is common in gauge theories, e.g., \cite{haag}, and postulated 
in, e.g.,
\cite{lbm, lmb}
in the context of quantum reference frames, theoretical quantities which do not commute with a symmetry
action are unobservable, even in principle.
The tension between this apparent unobservability and the use of such quantities in the description
of real physical systems is relieved by noting that any system of interest is 
a part of a larger whole - there is another system (called the environment, 
ancilla, reservoir, reference frame, or apparatus) whose presence is often assumed only implicitly
and which does not appear in the formulation. It is then possible that unobservable
quantities of the system of interest may be re-interpreted as representatives of 
observable {\it relative} quantities of system-plus-environment.
\par
In this paper we study the extent to which the same kind of restrictions (due to symmetry) hold in 
a dynamical context. To make the problem concrete, 
we must identify, in analogy to the invariance requirement
for observables, the right notion of restriction for channels, since it is not unique. This 
will be discussed in detail in a future publication.
One natural example arises in the 
presence of conserved quantities. 
The Wigner-Araki-Yanase (WAY) theorem \cite{ew1, ay1,lb1,lb2,lov1,mi} shows that 
 conservation laws restrict the accuracy and repeatability properties of a class of quantum measurements. 
Also in the presence of 
an additive conserved quantity,
\r{A}berg \cite{ab1} has shown that any unitary dynamics can be 
approximately realized by preparing ``coherent" (thus asymmetric) states of the environment.  
We here emphasize the necessity of the highly asymmetric (coherent) state 
(cf.\ \cite{lbm, lmb, mlb}), which requires a large environment in a sense to be described. 
If the environment is not large, states cannot have 
enough coherence and the possible dynamics is 
restricted. In addition, if the symmetry is not Abelian, 
there may be some restriction due to the uncertainty relation, because
large coherence with respect to some observable
 implies small coherence with respect to its conjugate. 

Another possible symmetry constraint is 
covariance of the dynamics, whose relevance 
to the reference frame context will be discussed 
briefly in the next section. 
As will be shown later, this constraint 
is weaker than the one given by the presence of 
conserved quantities. 
In this paper, we study the possible dynamics 
of the system under the 
symmetry constraint 
on the whole system (object system plus reference). 
We consider a quantum channel on the system and study how well 
this (target) channel is approximately realized by a 
covariant channel 
on the whole system, contingent upon a choice of state of 
the environment. We derive a quantitative relation which shows 
that for the covariant channel to be well approximated by 
the target channel, high asymmetry/coherence is required for the 
state of the environment. 
As an example we apply our relation to a qubit system under 
$SU(2)$ symmetry. 

Such size-versus-inaccuracy trade-offs are already present in the literature in various different contexts (see, e.g., \cite{mlb,ajr1,bartfin}), and our findings are broadly in line
with other findings - that good accuracy needs large size. Specific mention must be given to \cite{taj1,taj2,taj3}, which
has already investigated the dynamical setting and some elegant bounds have been provided, particularly in the unitary case.
However, we provide a novel quantitative bound in the dynamical context. 

 From a technical point of view, the present paper may be regarded as a descendant of \cite{mlb}, in which 
 we derived a quantitative bound in approximating an arbitrary effect by  a
 globally invariant effect. 
 There we reinterpreted the issue as an 
approximate 
joint measurement problem of observables and employed 
a quantitative uncertainty relation \cite{MiyaIma}. 
In this paper, we show that a similar technique, 
which we may call uncertainty relation based method, 
can be applied also to the approximating channel problem.

\section{Symmetry constraints on channels}
The principle of symmetry limits the possible observables to invariant ones (see, e.g., \cite{lmb}). 
There are some different generalizations 
of this constraint to channels. 
\par
Suppose that we have a system described by 
a Hilbert space $\hi$. 
By $\mathbf{B}(\hi)$ 
we denote the algebra of bounded operators on 
$\hi$.  
Throughout this paper every Hilbert space we encounter will be finite dimensional. 
We assume that 
a finite dimensional connected Lie group $G$ 
is acting on $\hi$. 
$G$ is assumed to define a true smooth unitary representation on each system, denoted by 
$U(g)$.
\par
One of the possible constraints on dynamics is 
given by a conservation law. 
Suppose that there exists a conserved 
charge 
$N$, 
which generates a $U(1)$ action. 
In the situation that the 
system is closed/isolated, the possible dynamics $\Lambda$
must satisfy $\Lambda(N)=N$ (in the Heisenberg 
picture). 
We thus arrive at the following definition. 
\begin{definition}
\emph{Invariant channels} are
defined as those $\Phi: \bhh \to \bhh$ for which $\Phi(U(g)) = U(g) \text{~for all~} g \in G$,
where $U$ is an $\hi$-representation of $G$. 
\end{definition}
\if
On the other hand, 
suppose that we have a particle and are asked to perform a repeatable measurement 
of the $z$-component of angular momentum. The dynamics may be
described by a L\"{u}ders channel 
$\Lambda(X)= \sum_{m} E_m X E_m$, where 
$L_z= \hbar \sum_m m  E_m$ is the spectral
decomposition. 
This channel implicitly assumes the existence of 
a reference frame/system which specifies the $z$-axis. 
If we employ another reference frame, 
what we measure is $U(R) L_z U(R)^*$ with 
some $R \in SO(3)$ and the corresponding 
channel becomes 
$\Lambda_R(X)= \hbar \sum_m U(R) E_m U(R)^*  
X U(R) E_m U(R)^*$. 
It satisfies $\Lambda_R(U(R)X U(R)^*)
=U(R) \Lambda (X) U(R)^*$. 
That is, $\Lambda_R(Y) = U(R) \Lambda
(U(R)^* Y U(R)) U(R)^*$ holds. 
If we are not informed which reference frame 
is to be used, we may choose one randomly. 
In this case a channel is described as 
\begin{eqnarray}
\overline{\Lambda}(Y) 
= \int \mu (dR) \Lambda_{R}(Y)
= \int \mu(dR) U(R) \Lambda(U(R)^*Y 
U(R))U(R)^*, 
\end{eqnarray}
where $\mu(\cdot)$ is (the) Haar measure on $SO(3)$.
\fi
On the other hand, 
suppose that we have a system and 
are asked to rotate the state around the $z$-axis 
by some angle $\theta$. The desired 
channel is $\Phi (X) 
= e^{i S_z \theta}
X e^{-i S_z \theta}$, 
where $S_z$ is the $z$-component of 
angular momentum. (We work in units in which $\hbar =1$.)  
This channel implicitly assumes the existence of 
a reference frame/system which specifies the $z$-axis. 
If we employ another reference frame, 
the $z$-component of 
angular momentum is represented as 
$U(R) S_z U(R)^*$ with 
some $R \in SO(3)$ and the corresponding 
channel becomes 
$\Phi_R(X) 
= U(R) \Phi(U(R)^* X U(R))U(R)^*$. 
 If we are not informed which reference frame 
is to be used, we may choose one randomly, in which 
 case the channel is described as 
\begin{eqnarray}
\overline{\Phi}(X)
=
\int \mu(dR) \Phi_R(X), 
\end{eqnarray}
where $\mu(\cdot)$ is (the) Haar measure on $SO(3)$.
$\overline{\Phi}$ is an example of a \emph{covariant channel}
which we now define.
\begin{definition}
A channel $\Lambda : \bhh \to \bhh$ is called 
covariant if and only if 
\begin{eqnarray}
\Lambda(U(g)^* X U(g))= U(g)^* \Lambda (X) U(g)
\label{covchannel}
\end{eqnarray}
holds for all $X \in \mathbf{B}(\mathcal{H})$ 
and for all $g\in G$. 
\end{definition}
As the next proposition shows, invariant channels form an important subclass of covariant channels.
\begin{proposition}\label{prop:p1}
Invariant channels are covariant.
\end{proposition}
\begin{proof}
We begin by using a channel $\Lambda$ to define an ``operator-valued inner product"
$\langle\langle A|B \rangle\rangle:=
\Lambda(A^*B)- \Lambda(A)^*\Lambda(B)$, 
which satisfies a Cauchy-Schwarz type inequality (see \cite{janssens} and Lemma 3 in \cite{mlb}):
\begin{eqnarray}\label{eq:C-S}
\Vert \langle \langle A|B\rangle \rangle \Vert^2
\leq \Vert \langle \langle A|A\rangle \rangle\Vert 
\Vert \langle \langle B |B \rangle \rangle\Vert,
\end{eqnarray}
where $|| \cdot ||$ denotes the standard operator norm in $\bhh$.
Suppose that a unitary $U$ is a fixed point, i.e., $\Lambda(U) = U$. 
Then it holds that 
\begin{eqnarray*}
\langle \langle U| U \rangle \rangle
= \id - \Lambda(U)^* \Lambda(U)= 0 .
\end{eqnarray*}
Thus for such a $U$ and arbitrary $A$ we find
\begin{eqnarray*}
\langle \langle A|U\rangle \rangle =\Lambda(A^*U) -\Lambda(A^*)\Lambda(U)=
\Lambda(A^*U)-\Lambda(A)^* U=0, {~\text{by}~ \eqref{eq:C-S}}. 
\end{eqnarray*}
Now Let $\Lambda: \bhh \to \bhh$ be an invariant channel. Then for all $g \in G$,
\begin{eqnarray*}
 \Lambda(AU(g))= \Lambda(A)\Lambda(U(g))
= \Lambda(A)U(g).
\end{eqnarray*}
Similarly $\Lambda(U(g)^* B) = U(g)^* \Lambda(B)$, and thus $\Lambda(U(g)^* A U(g))= U(g)^* \Lambda(A) U(g).$
\end{proof}

We note that there exist covariant channels 
which are not invariant; for instance, for any 
invariant state 
$\omega_0$, the channel 
$\Lambda(X) = \omega_0(X) \id$ is covariant but not invariant. 
However, covariance and invariance are equivalent for unitary channels:
\begin{proposition}
Unitary covariant channels are invariant. 
\end{proposition}
\begin{proof}
Let us consider a unitary (and thus automorphic), covariant channel $\Lambda : \bhh \to \bhh$, i.e., 
$\Lambda(X) = V^*XV$ and $U(g)^* \Lambda (X) U(g) = \Lambda (U(g)^*XU(g))$.
Then it holds that  for all $X$
\begin{eqnarray*}
\Lambda(U(g)^* XU(g))
= V^* U(g)^* V V^* X V V^* U(g) V
= U(g)^* (V^* X V) U(g). 
\end{eqnarray*}
Now putting 
$V^*XV=Z$, we apply $U(g) \cdot U(g)^*$ to the above equation to obtain
\begin{eqnarray*}
(U(g)V^* U(g)^* V)Z (V^* U(g) V U(g)^*)=Z,
\end{eqnarray*}
which implies 
$V^*U(g) VU(g)^* = c(g) \id$ with $|c(g)|=1$ and 
$V^* U(g) V = c(g) U(g)$. As the left-hand side 
satisfies $V^*U(g)V V^*U(g')V =V^*U(gg') V$, 
$c(g)c(g')=c(gg')$ holds for all $g, g\in G$. 
Now in a neighborhood of $e\in G$, 
for each element $l$ of Lie algebra the corresponding generator 
$L$ exists satisfying $U(e^{ls})= e^{iL s}$ for sufficiently small $|s|$. 
If we put the generator of $V^*U(e^{ls})V$ as $L'$, 
it satisfies $V^*LV=L'$. It in addition satisfies $L' = L+ k\id$ 
for some $k\in \mathbf{R}$ as $V^*U(e^{ls})V= c(e^{ls}) U(e^{ls})$ must hold. 
But as $L$ is bounded (as $\mathcal{H}$ is finite dimensional) and $\Vert L\Vert = \Vert L'\Vert $ holds, 
$k=0$ is the only possible choice. 
Thus we have shown that for a neighbourhood $N_e$ of $e\in G$ 
$V^* U(g) V= U(g)$ is satisfied. 
As $G$ is connected, it is generated by 
$N_e$. It implies that $c(g)=1$ for all $g\in G$.  
\end{proof}
\section{The setting and results}
As we have seen in the last section, 
we cannot implement (for instance) the rotation 
around the $z$-axis without 
using a ``correct" 
reference frame. More precisely, we may  
implement the right rotation but this occurs 
only by chance. The averaged channel is 
a covariant $\overline{\Phi}$ which is different from 
the desired rotation. In the worst case, 
the discrepancy is larger than the averaged case.  
Thus we must have a reference frame. 
Since a reference frame is also a physical system, 
there should be a quantum description. Our question is 
then to ask what is the condition on the 
quantum reference frame so that it works 
well to implement the desired channel. 
In the following we formulate the problem 
in a general setting.  
\par
Let $G$ be a connected Lie group. 
We have a system and a reference frame 
described by (as always, finite dimensional) Hilbert spaces $\his$ and $\hir$, and 
on each space, $G$ has a smooth true unitary 
representation $U_{\Sy}(g)$ and $U_{\R}(g)$. 
Their composition is written as 
$U(g)=U_{\Sy}(g) \otimes U_{\R}(g)$ 
which acts on $\hi=\hi_{\Sy}\otimes \hi_{\R}$. 
Our purpose is to study how well a general channel $\Lambda: \bhs \to \bhs$ 
is approximately realized by the restriction of 
a covariant channel 
$\Phi: \bht \to \bht$. We view $\Phi$ as representing the ``true" transformation, 
with its restriction representing the transformation with the additional system suppressed.
Therefore, 
$\Phi$ satisfies 
\begin{eqnarray*}
\Phi(U(g)^* X U(g))=U(g) \Phi(X) U(g)
\end{eqnarray*}
for all $g\in G$ and $X \in \bht$.
On the level of observables, the restriction to the system 
$\Gamma_{\rho_{\R}}: \bht \to \bhs$ is determined by a state $\rho_{\R}$ on $\bhr$ 
and is defined by the completely positive conditional expectation
\begin{eqnarray*}
\mbox{tr}[\rho_{\Sy} \Gamma_{\rho_{\R}}(X)]
= \mbox{tr}[(\rho_{\Sy} \otimes \rho_{\R})X],
\end{eqnarray*}
which holds for all states $\rho_{\Sy}$ of the system and $X \in \bht$. 
In order to define the restriction for channels, 
we use the natural inclusion $\iota: \bhs \to \bht$, given as 
\begin{eqnarray*}
\iota(A) = A\otimes \id_{\R}. 
\end{eqnarray*}
Then the realized channel is written as 
$\Phi_{\rho_{\R}}:=\Gamma_{\rho_{\R}} \circ \Phi\circ \iota: \bhs\to \bhs$,
and we wish to quantify the discrepancy between $\Phi_{\rho_{\R}}$ and $\Lambda$.
As a quantity to characterize the discrepancy, 
one may employ the norm difference between 
two channels defined by 
\begin{eqnarray*}
\Vert \Phi_{\rho_{\R}} - \Lambda\Vert_{\mathrm{channel}} 
:= \sup_{X\in \mathbf{B}(\mathcal{H}_S), 
\Vert X\Vert =1}
\Vert \Phi_{\rho_{\R}}(X) - \Lambda(X)\Vert. 
\end{eqnarray*}
\par
For each element of the Lie algebra $\frak{g}$ 
of a Lie group $G$ there exists a corresponding self-adjoint operator (the generator) 
 acting in $\his$.
For each $l\in \frak{g}$, there exist operators 
$L_{\Sy}$ and $L_{\R}$ satisfying $U_{\Sy}(e^{l s}) 
= e^{i L_{\Sy} s}$ and $U_{\R}(e^{ls})= e^{i L_{\R} s}$ 
and therefore $U(e^{ls})= e^{i (L_{\Sy} \otimes \id_{\R} + \id_{\Sy} \otimes L_{\R})}$. 
As unitary operators have norm $1$, we obtain 
an inequality for each $U_{\Sy}(e^{l s_0})$, 
\begin{eqnarray*}
\Vert \epsilon(l:s_0)
\Vert := 
\Vert 
\left(\Gamma_{\rho_{\R}}\circ \Phi \circ \iota\right) (U_{\Sy}(e^{l s_0}))
- \Lambda(U_{\Sy}(e^{l s_0}))
\Vert 
\leq 
\Vert \Phi_{\rho_{\R}} - \Lambda\Vert_{\mathrm{channel}}.
\end{eqnarray*}

$F(\rho_0, \rho_1)$ represents the fidelity between two states $\rho_0$ and 
$\rho_1$ defined by $F(\rho_0, \rho_1):= \mbox{tr}[\sqrt{\rho_0^{1/2}\rho_1
\rho_0^{1/2}}]$. This quantity is positive and equals  $1$ if and only if 
$\rho_0= \rho_1$ holds. 
\begin{theorem}\label{th:main}
Let $L_{\Sy}$ and $L_{\R}$ be generators of unitary representations 
of $e^{ls}\ (s\in \mathbb{R})$ on $\hs$ and $\hr$ for $l\in \frak{G}$. 
Define $U_{\Sy}(l:s):=e^{i L_{\Sy} s}$ and $U_{\R}(l:s) = e^{i L_{\R} s}$. 
Then, for any $s_0 \in \mathbb{R}$, 
 $\epsilon(l:s_0):= \left(\Gamma_{\rho_{\R}}\circ \Phi \circ \iota\right) (U_{\Sy}(l:s_0))
- \Lambda(U_{\Sy}(l:s_0))$ is bounded for all $s\in \mathbb{R}$ by:
\begin{align*}
 & \Vert
[\Lambda(U_{\Sy}(l:s_0)), U_{\Sy}(l:s)]
\Vert
\leq
2 \Vert U_{\Sy}(l:s)- \id\Vert \Vert \epsilon(l:s_0) \Vert 
\\
&
+
\left(\frac{1}{F(\rho_{\R}, U_{\R}(l:s) \rho_{R}U_{\R}(l:s)^*)^2}
-1\right)^{1/2}
\left(
\left(
\Vert \id_{\Sy} - \Lambda(U_{\Sy}(l:s_0))^* \Lambda(U_{\Sy}(l:s_0))\Vert 
+ 2 \Vert \epsilon(l:s_0)\Vert\right)^{1/2} 
\right. 
\\
&
\left.
+
\left(
\Vert \id_{\Sy} - \Lambda(U_{\Sy}(l:s_0))\Lambda(U_{\Sy}(l:s_0))^*\Vert 
+2 \Vert \epsilon(l:s_0)\Vert\right)^{1/2}
\right).
\end{align*}

\end{theorem}
Before proving Theorem \ref{th:main}, we present some immediate implications.
We first observe that the left hand side of the above inequality 
vanishes for covariant $\Lambda$, since
\begin{eqnarray}
\Vert [\Lambda(U_{\Sy}(l:s_0)), U_{\Sy}(l:s)]\Vert
&=& \Vert U_{\Sy}(l:s)^* [\Lambda(U_{\Sy}(l:s_0)), U_{\Sy}(l:s)]\Vert
\nonumber
\\
&=&
\Vert U_{\Sy}(l:s)^* \Lambda(U_{\Sy}(l:s_0))
U_{\Sy}(l:s) - \Lambda(U_{\Sy}(l:s_0))\Vert,  
\label{eq1}
\end{eqnarray} 
and
\begin{eqnarray*}
U_{\Sy}(l:s)^* \Lambda(U_{\Sy}(l:s_0))
U_{\Sy}(l:s) = \Lambda( U_{\Sy}(l:s)^* U_{\Sy}(l:s_0)
U_{\Sy}(l:s))
= \Lambda(U_{\Sy}(l:s_0)). 
\end{eqnarray*}

Therefore, there is no bound for approximating covariant channels
$\mathbf{B}(\his) \to \mathbf{B}(\his)$ by restrictions of covariant channels
$\mathbf{B}(\hi) \to \mathbf{B}(\hi)$. Indeed, any covariant $\Lambda$ can trivially
be written as the restriction of a covariant channel $\Phi$ on $\mathbf{B}(\hi)$, i.e., as 
$\Phi_{\rho_{\R}}$ for all $\rho_{\R}$, by setting $\Phi = \Lambda \otimes \mbox{id}$.
If $\Lambda$ is a unitary channel, Theorem \ref{th:main} takes a much simpler form.
\begin{corollary}
Under the same assumptions as Theorem \ref{th:main}, but for unitary $\Lambda$, 
it holds that 
\begin{align*}
\Vert
[\Lambda(U_{\Sy}(l:s_0)), U_{\Sy}(l:s)]
\Vert \leq
&2 \Vert U_{\Sy}(l:s)- \id\Vert \Vert \epsilon(l:s_0) \Vert 
\\
+& 2 
\left(\frac{1}{F(\rho_{\R}, U_{\R}(l:s) \rho_{R}U_{\R}(l:s)^*)^2}
-1\right)^{1/2} 
\Vert \epsilon(l:s_0)\Vert^{1/2}.
\end{align*}
\end{corollary}

The proof follows from the observation that if $\Lambda$ is a unitary channel (or indeed, multiplicative), then for any unitary operator $U\in \bhs$, we have $\Lambda(U)^* \Lambda(U) = \id$.  

Therefore, we see that in order to make possible good agreement between $\Lambda$ and $\Phi_{\rho_{\R}}$, a highly ``asymmetric" reference state $\rho_{\R}$ is necessary, since
$F(\rho_{\R}, U_{\R}(l:s) \rho_{\R}U_{\R}(l:s)^*)$ must decrease rapidly 
with respect to $|s|$ as otherwise the left-hand side of the inequality 
can be large for non-covariant $\Lambda$. 
\par
Furthermore, this asymmetry, or \emph{coherence} factor, can be 
bounded by the ``spread" of the (symmetry) generator $L_{\R}$: 

\begin{corollary}\label{cor:der}
In the same scenario as Theorem \ref{th:main}, it holds that  
\begin{align*}
\Vert
[\Lambda(U_{\Sy}(l:s_0)), L_{\Sy}]
\Vert
\leq
& 2 \Vert L_{\Sy} \Vert \Vert \epsilon(l:s_0) \Vert \\
+
&(\Delta_{\rho_{\R}}L_{\R})
\biggl( (
\Vert \id_{\Sy} - \Lambda(U_{\Sy}(l:s_0))^* \Lambda(U_{\Sy}(l:s_0))\Vert 
+ 2 \Vert \epsilon(l:s_0)\Vert )^{1/2} \\
+ &
(
\Vert \id_{\Sy} - \Lambda(U_{\Sy}(l:s_0))\Lambda(U_{\Sy}(l:s_0))^*\Vert 
+2 \Vert \epsilon(l:s_0)\Vert )^{1/2}
\biggr),
\end{align*}
where $\Delta_{\rho_{\R}}L_{\R}
:= \sqrt{\mathrm{tr}[\rho_{\R} L_{\R}^2]- \mathrm{tr}[\rho_{\R}L_{\R}]^2}$ 
represents the standard deviation of $L_{\R}$ in the state $\rho_{\R}$.
\end{corollary}
\begin{corollary}
Under the same assumptions as Theorem \ref{th:main}, but for unitary $\Lambda$, 
it holds that 
\begin{eqnarray*}
&&\Vert
[\Lambda(U_{\Sy}(l:s_0)), L_{\Sy}]
\Vert
\leq
2 \Vert L_{\Sy} \Vert \Vert \epsilon(l:s_0) \Vert 
+2\sqrt{2}
(\Delta_{\rho_{\R}}L_{\R})
 \Vert \epsilon(l:s_0)\Vert^{1/2}. 
\end{eqnarray*}
\end{corollary}
This immediately follows from Corollary 
\ref{cor:der}. 
The inequality is easy to interpret. 
For non-covariant $\Lambda$ which 
yields non-vanishing left-hand side, 
$\Delta_{\rho_{\R}}L_{\R}$ must be large 
to attain small $\Vert \epsilon(l:s_0)\Vert$. 
Thus it implies that the reference system $\R$ 
must be large (macroscopic). This result has some qualitative similarity to 
the bounds obtained in \cite{taj1,taj2}, where large size/coherence/energy fluctuation of the reference is shown to be necessary for implementing unitary dynamics.

%
\par
We now present proofs of Theorem \ref{th:main} and Corollary \ref{cor:der}.
To prove Theorem \ref{th:main}, we need the following lemma \cite{mlb, janssens}:
\begin{lemma}\label{th:uncertain}
Consider a channel $\Gamma:\bhh \to \mathbf{B}(\hik)$ for 
Hilbert spaces $\hi$ and $\hik$. 
If $A,B\in \bhh$ satisfy $[A,B]=0$, then
\begin{align}\label{eq:incomprehensible_inequality}
\Vert
 [\Gamma(A),\Gamma(B)]
\Vert
\leq
&\Vert 
\Gamma(A^*A)-\Gamma(A)^*\Gamma(A)\Vert^{1/2}
\Vert
\Gamma(BB^*)-\Gamma(B)\Gamma(B)^*\Vert^{1/2}
\\
+
& \Vert 
\Gamma(AA^*)-\Gamma(A)\Gamma(A)^*\Vert^{1/2}
\Vert
\Gamma(B^*B)-\Gamma(B)^*\Gamma(B)\Vert^{1/2}.
\end{align}  
\end{lemma}
We now present the proof of Theorem \ref{th:main}.
\begin{proof} 
If the state $\rho_{\R}$ and $s \in \mathbb{R}$ satisfy
 $F(\rho_{\R}, U_{\R}(l:s) \rho_{\R}
U_{\R}(l:s)^*)=0$, the claim follows trivially, and thus we assume otherwise. 
For notational simplicity, we omit the dependence on $l$ and 
write $U_{\Sy}(s)$ for $U_{\Sy}(l:s)$,
$U_{\R}(s)$ for $U_{\R}(l:s)$ 
 and $\epsilon(s_0)$ for $\epsilon(l:s_0)$.  
We first write 
\begin{eqnarray}\label{eq:bcom}
[\Lambda(U_{\Sy}(s_0)), U_{\Sy}(s)]
=[U_{\Sy}(s), \epsilon(s_0)] + [\Gamma_{\rho_{\R}}\Phi(U_{\Sy}(s_0) \otimes \id_{\R}), U_{\Sy}(s)].  
\end{eqnarray}
The first term on the right hand side is bounded as
\begin{eqnarray*}
\Vert [U_{\Sy}(s), \epsilon(s_0)]\Vert 
= \Vert[U_{\Sy}(s)-\id, \epsilon(s_0)]\Vert
\leq 2 \Vert U_{\Sy}(s)- \id\Vert \Vert \epsilon(s_0) \Vert. 
\end{eqnarray*}
To estimate the second term on the right hand side of \eqref{eq:bcom},
we introduce a purification of $\rho_{\R}$ to 
$|\phi_{RZ}\rangle \in \hir \otimes \hi_Z$, where we choose the  
purification space $\hi_Z$ to be minimal, i.e., its dimension coincides with 
the rank of $\rho_{\R}$. 
We denote $\Gamma_{|\phi_{RZ}\rangle \langle \phi_{RZ}|}
: \mathbf{B}(\his \otimes \hir \otimes \hi_Z) 
\to \mathbf{B}(\his)$ by $\Gamma$ for simplicity. 
Now, for an arbitrary operator $W_Z$ on $\hi_Z$, we have
\begin{eqnarray*}
\Gamma(U_{\Sy}(s) \otimes U_{\R}(s) \otimes W_Z) 
= U_{\Sy}(s) \langle \phi_{RZ}| U_{\R}(s) \otimes W_Z|\phi_{RZ}\rangle. 
\end{eqnarray*}
In the following 
we denote $\Phi \otimes \mbox{id}_{Z}$ by 
$\hat{\Phi}$, and in order to simplify some long expressions we will make the abbreviations $A_0 = U_{\Sy}(s_0)\otimes \id _{\R} \otimes \id _Z$ and $A_s = U_{\Sy}(s)\otimes U_{\R}(s)\otimes W_Z$ when convenient. 
Thus we have, for $W_Z$ with $ \langle \phi_{RZ} |U_{\R}(s) \otimes W_Z|\phi_{RZ}\rangle
\neq 0$, 
\begin{equation*}
[\Gamma (\hat{\Phi}(A_0)), 
U_{\Sy}(s)]
= 
\frac{[\Gamma(\hat{\Phi}(A_0)), 
\Gamma(A_s)]}{\langle \phi_{RZ} |U_{\R}(s) \otimes W_Z|\phi_{RZ}\rangle }. 
\label{eqn1234}
\end{equation*}
Since $\Phi$ is a covariant channel, it holds that 
\begin{eqnarray*}
(U_{\Sy}(s)^* \otimes U_{\R}(s)^*)\Phi(U_{\Sy}(s_0)\otimes \id_{\R})
(U_{\Sy}(s) \otimes U_{\R}(s))
= \Phi (U_{\Sy}(s_0) \otimes \id_{\R}). 
\end{eqnarray*}
Therefore we find 
\begin{eqnarray*}
[\Phi( U_{\Sy}(s_0) \otimes \id_{\R}) \otimes \id_Z, 
U_{\Sy}(s) \otimes U_{\R}(s) \otimes W_Z]=0,  
\end{eqnarray*}
which enables us to apply Lemma \ref{th:uncertain}. 
In the following, $W_Z$ is chosen to be unitary.
Now we bound
\begin{equation*}
\Vert [\Gamma (\hat{\Phi}(A_0)), 
U_{\Sy}(s)]\Vert \\
=
\frac{\Vert [\Gamma(\hat{\Phi}(A_0)), 
\Gamma(A_s)]\Vert}{| \langle \phi_{RZ} |U_{\R}(s) \otimes W_Z|\phi_{RZ}\rangle | }. 
\end{equation*}
Then 
Lemma \ref{th:uncertain} yields 
the numerator of the above equation to be  
bounded as  
\begin{align*}
\Vert [\Gamma (\hat{\Phi}(A_0)), \Gamma(A_s)]
\Vert 
\leq &\Vert \Gamma (\hat{\Phi}(A_0)^*\hat{\Phi}
(A_0))
 - \Gamma( \hat{\Phi}(A_0))^*\Gamma(\hat{\Phi}(A_0))\Vert^{1/2}
\Vert
\id - \Gamma(A_s) \Gamma( A_s)^* \Vert^{1/2} \\
+  &\Vert  \Gamma(\hat{\Phi}(A_0)
\hat{\Phi}(A_0)^*)
 - \Gamma( \hat{\Phi}(A_0))\Gamma(\hat{\Phi}
(A_0))^* \Vert^{1/2} 
\Vert 
\id - \Gamma(A_s)^*
\Gamma( A_s) \Vert^{1/2}.
\end{align*}
We first estimate the norm of  
\begin{equation*}
A:= 
\Gamma(\hat{\Phi}(A_0)^*\hat{\Phi}
(A_0)) - \Gamma( \hat{\Phi}(A_0)^*\Gamma(\hat{\Phi}(A_0)).
\end{equation*}
Due to the two-positivity of 
$\Gamma$ (i.e., $\Gamma(X^*X) \geq \Gamma(X)^*
\Gamma(X)$ for all $X$) the operator $A$ is 
positive. 
Furthermore applying the two-positivity of $\hat{\Phi}$, 
we obtain 
\begin{equation*}
\hat{\Phi}(A_0)^* \hat{\Phi}(A_0)
\leq \hat{\Phi}(A_0^*
A_0)
=\id.
\end{equation*}
Since $\Gamma$ is a positive map we find 
\begin{eqnarray*}
\mathbf{0}\leq A \leq \id 
 - \Gamma( \hat{\Phi}(U_{\Sy}(s_0) \otimes \id_{\R}
\otimes \id_Z))^*\Gamma(\hat{\Phi}(U_{\Sy}(s_0) \otimes \id_{\R}
\otimes \id_Z)),
\end{eqnarray*}
from which we conclude 
\begin{eqnarray*}
\Vert A\Vert 
\leq \Vert \id 
 - \Gamma( \hat{\Phi}(U_{\Sy}(s_0) \otimes \id_{\R}
\otimes \id_Z))^*\Gamma(\hat{\Phi}(U_{\Sy}(s_0) \otimes \id_{\R}
\otimes \id_Z))\Vert. 
\end{eqnarray*}
The term 
\begin{equation*}
 \Vert 
\Gamma(\hat{\Phi}(A_0)
\hat{\Phi}(A_0)^*)
 - \Gamma( \hat{\Phi}(A_0))\Gamma(\hat{\Phi}
(A_0))^* \Vert^{1/2} 
\end{equation*}
can be treated similarly. Writing $c_{RZ} \equiv \langle \phi_{RZ} |U_{\R}(s) \otimes W_Z|\phi_{RZ}\rangle$, we thus obtain 

\begin{multline*}
\Vert [\Gamma (\hat{\Phi}(A_0)), U_{\Sy}(s)]\Vert 
\leq  \frac{1}{| c_{RZ}| }
\biggl(
\Vert \id - \Gamma( \hat{\Phi}(A_0))^*\Gamma(\hat{\Phi}(A_0))\Vert^{1/2}
\Vert
\id - \Gamma(A_s)
\Gamma( A_s)^* \Vert^{1/2}\\ + 
\Vert \id - \Gamma( \hat{\Phi}(A_0))\Gamma(\hat{\Phi}(A_0))^* \Vert^{1/2}
\Vert 
\id - \Gamma(A_s)^*
\Gamma( A_s) \Vert^{1/2} \biggr) ,
\end{multline*}
which is bounded above by
\begin{multline*}
\frac{1}{|c_{RZ} | }
( 1- |c_{RZ}|^2)^{1/2} \biggl(
\Vert \id - \Gamma_{\rho_{\R}}( \Phi(U_{\Sy}(s_0) \otimes \id_{\R}))^*\Gamma_{\rho_{\R}} (\Phi(U_{\Sy}(s_0) \otimes \id_{\R})) \Vert^{1/2}\\
 + 
\Vert \id - \Gamma_{\rho_{\R}}( \Phi(U_{\Sy}(s_0) \otimes \id_{\R}))\Gamma_{\rho_{\R}}(\Phi(U_{\Sy}(s_0) \otimes \id_{\R}))^* \Vert^{1/2}
\biggr).
\end{multline*}
We estimate 
\begin{align*}
\Vert \id - &\Gamma_{\rho_{\R}}( \Phi(U_{\Sy}(s_0) \otimes \id_{\R}))^*\Gamma_{\rho_{\R}} (\Phi(U_{\Sy}(s_0) \otimes \id_{\R})) \Vert \\
&=\Vert \id - \Lambda(U_{\Sy}(s_0))^* \Lambda(U_{\Sy}(s_0)) 
- \epsilon(s_0)^* \epsilon (s_0) - \epsilon (s_0)^* \Lambda(U_{\Sy}(s_0))- 
\Lambda(U_{\Sy}(s_0))^* \epsilon (s_0)\Vert  \\
&\leq \Vert \id - \Lambda(U_{\Sy}(s_0))^* \Lambda(U_{\Sy}(s_0))\Vert 
+ 2 \Vert \epsilon(s_0)\Vert. 
\end{align*}
Similarly we obtain
\begin{align*}
\Vert \id - &\Gamma_{\rho_{\R}}( \Phi(U_{\Sy}(s_0) 
\otimes \id_{\R}))\Gamma_{\rho_{\R}}(\Phi(U_{\Sy}(s_0) \otimes \id_{\R}))^* \Vert
\\
&\leq 
\Vert \id - \Lambda(U_{\Sy}(s_0))\Lambda(U_{\Sy}(s_0))^*\Vert 
+2 \Vert \epsilon(s_0)\Vert.
\end{align*}
Finally, one can choose $W_Z$ so as to maximize 
$|\langle \phi_{RZ}| U_{\R}(s) \otimes W_Z |\phi_{RZ}\rangle|$,
which coincides with $F(\rho_{\R}, U_{\R}(s) \rho_{R}U_{\R}(s)^*)$ due to 
Uhlmann's theorem \cite{uhl1}, thereby completing the proof.
\end{proof}
We now provide a proof of Corollary \ref{cor:der}.
\begin{proof}
Adopting the shorthand $F \equiv F(\rho_{\R}, U_{\R}(l:s) \rho_{R}U_{\R}(l:s)^*)$, the equality (\ref{eq1})
replaces the inequality of Theorem \ref{th:main} by, 
\begin{multline*}
\Vert U_{\Sy}(l:s)^* \Lambda(U_{\Sy}(l:s_0))
U_{\Sy}(l:s) - \Lambda(U_{\Sy}(l:s_0))\Vert
\leq
2 \Vert U_{\Sy}(l:s)- \id\Vert \Vert \epsilon(l:s_0) \Vert 
\\
+
\bigl(\frac{1}{F^2}
-1\bigr)^{1/2}
\left(
\left(
\Vert \id_{\Sy} - \Lambda(U_{\Sy}(l:s_0))^* \Lambda(U_{\Sy}(l:s_0))\Vert 
+ 2 \Vert \epsilon(l:s_0)\Vert\right)^{1/2} 
\right.
\\
\left.
+
\left(
\Vert \id_{\Sy} - \Lambda(U_{\Sy}(l:s_0))\Lambda(U_{\Sy}(l:s_0))^*\Vert 
+2 \Vert \epsilon(l:s_0)\Vert\right)^{1/2}
\right).
\end{multline*}
To bound the first term on the right hand side we write 
\begin{eqnarray*}
U_{\Sy}(s) = \id_{\Sy} +i \int^s_0 dt U_{\Sy}(t) L_{\Sy} , 
\end{eqnarray*}
and therefore
\begin{eqnarray*}
\Vert U_{\Sy}(s) - \id_{\Sy}\Vert
\leq |s| \Vert L_{\Sy}\Vert. 
\end{eqnarray*}
For the second term, we bound $F(\rho_{\R}, U_{\R}(l:s) \rho_{\R}U_{\R}(l:s)^*)$ by choosing a purification of $\rho_{\R}$ as $|\phi\rangle \in \hir \otimes \hi_Z$. 
Then Uhlmann's theorem states that the fidelity is written as 
\begin{eqnarray*}
F(\rho_{\R}, U_{\R}(l:s) \rho_{\R} U_{\R}(l:s)^*)
= \sup_{|\phi\rangle} |\langle \phi | e^{i L_{\R} s} \otimes \id_Z|\phi\rangle|. 
\end{eqnarray*}
For each purification $|\phi\rangle$, the Mandelstam-Tamm uncertainty 
relation \cite{mt1,bu1} provides a bound for $0 \leq \Delta_{\rho_{\R}} L_{\R} \cdot s\leq \pi/2$, 
\begin{eqnarray*}
|\langle \phi | e^{i L_{\R} s} \otimes \id_Z |\phi \rangle |
\geq \cos (\Delta_{\rho_{\R}} L_{\R} \cdot s).  
\end{eqnarray*}
Thus we obtain 
\begin{eqnarray*}
\left(\frac{1}{F(\rho_{\R}, U_{\R}(l:s) \rho_{R}U_{\R}(l:s)^*)^2}
-1\right)^{1/2}
\leq \tan (\Delta_{\rho_{\R}} L_{\R} \cdot s).  
\end{eqnarray*}
We divide the both terms by $|s|$ and take $|s|\to 0$ to obtain, 
\begin{align*}
\Vert[ &\Lambda(U_{\Sy}(s_0)), L_{\Sy}] \Vert
\leq  2 \Vert L_{\Sy}\Vert \Vert \epsilon(s_0)\Vert 
 + \Delta_{\rho_{\R}} L_{\R}
\biggl (\Vert \id- \Lambda(U_{\Sy}(s_0))^* \Lambda(U_{\Sy}(s_0))\Vert 
+ 2 \Vert \epsilon (s_0)\Vert)^{1/2}\\
&+
  (\Vert \id- \Lambda(U_{\Sy}(s_0)) \Lambda(U_{\Sy}(s_0))^*\Vert 
+ 2 \Vert \epsilon (s_0)\Vert)^{1/2}
\biggr ).\qedhere
\end{align*}
\end{proof}


\section{Rotational symmetry}
As an example of the general behaviour we have investigated, we consider the possible dynamics 
of a qubit with Hilbert space $\his = \mathbb{C}^2$ under $SO(3)$ symmetry, realized by a true irreducible unitary representation of 
its universal covering group $SU(2)$. 
Since only a trivial unitary operator proportional to $\id$ 
commutes with all $SU(2)$ generators (angular momenta), 
one cannot change the state of the qubit in isolation 
(i.e., unitarily). 
The environment $\hir$ also has $SU(2)$ as a symmetry. 
We denote the angular momenta of the system and the reference frame 
by $s_j$ and $S_j$ $(j =x,y,z)$ respectively. 
We consider an $SU(2)$-covariant channel $\Phi: \mathbf{B}(\his \otimes \hir) 
\to \mathbf{B}(\his \otimes \hir)$. 
The following corollary is immediately obtained 
from Corollary \ref{cor:der}. 
\begin{corollary}
Let $G$ be a Lie group. 
For a covariant channel $\Phi:
\bht \to \bht$, its restriction $R \equiv \Gamma_{\rho_{\R}}\circ \Phi \circ \iota:
\bhs \to \bhs$, 
satisfies 
 \begin{equation*}
\Vert[ R (U_{\Sy}(s_0)), 
L_{\Sy}]\Vert 
\leq  (\Delta_{\rho_{\R}} L_{\R})
\biggl(
\Vert \id - R
(U_{\Sy}(s_0))^*R (U_{\Sy}(s_0))\Vert^{1/2} 
+
\Vert \id - R(U_{\Sy}(s_0))
R(U_{\Sy}(s_0))^*\Vert^{1/2}
\biggr).
\end{equation*}
\end{corollary}
We apply this to 
the case $G= SU(2)$ and $\his = \mathbb{C}^2$.
For $l = s_x$, 
$U_{\Sy}(s)$ is written as $U_{\Sy}(s)= e^{i s_x s} = e^{i \frac{\sigma_x}{2} s}$. 
We set $s_0= \pi$ to obtain $U_{\Sy}(s_0) = i \sigma_x$.  
Then the 
restriction to the system $\Lambda:=\Gamma_{\rho_{\R}}\circ \Phi \circ \iota : 
\mathbf{B}(\his) \to \mathbf{B}(\his)$ satisfies the following three 
inequalities:  
\begin{align*}
&\Vert[\Lambda(\sigma_x), s_x]\Vert 
\leq 2 (\Delta_{\rho_{\R}}S_x) \Vert \id - \Lambda(\sigma_x)^2\Vert^{1/2}\\
&\Vert[\Lambda(\sigma_y), s_y]\Vert 
\leq 2 (\Delta_{\rho_{\R}}S_y) \Vert \id - \Lambda(\sigma_y)^2\Vert^{1/2}\\
&\Vert[\Lambda(\sigma_z), s_z]\Vert 
\leq 2 (\Delta_{\rho_{\R}}S_z) \Vert \id - \Lambda(\sigma_z)^2\Vert^{1/2},
\end{align*}
where $s_x=\frac{1}{2}\sigma_x$, {\it etc}.
The uncertainty relations for angular momenta gives a non-trivial bound 
on sums of their fluctuations. We consider 
\begin{eqnarray*}
(\Delta S_x)^2 + (\Delta S_y)^2 + (\Delta S_z)^2
&=&
\langle S_x^2 + S_y^2 + S_z^2\rangle 
- (\langle S_x\rangle^2 + \langle S_y\rangle^2 
+\langle S_z\rangle^2)
\\
&\leq& l(l+1) -   (\langle S_x\rangle^2 + \langle S_y\rangle^2 
+\langle S_z\rangle^2),
\end{eqnarray*}
where $l$ is the magnitude of the largest spin of the environment. 
(Note that $\hr$ is written as a direct sum of 
irreducible representations of $SU(2)$ as
$\hr = \oplus_s \mathbb{C}^{2s+1}$. $l$ is the largest value of $s$ in 
the summation.)
It is easy to show that 
$\langle S_x\rangle^2 + \langle S_y\rangle^2 + \langle S_z\rangle^2$ 
is rotationally invariant. 
We consider the quantity $\langle \mathbf{S}\cdot \mathbf{n}\rangle$ 
for $|\mathbf{n}|=1$. 
This is a smooth function over the sphere and therefore has 
a maximum value at a certain point. 
To estimate the value of $\langle S_x\rangle^2 + \langle S_y\rangle^2 
+ \langle S_z\rangle^2$, we assume 
that the maximum of $\langle \mathbf{S}\cdot \mathbf{n}\rangle$ 
is attained at $\mathbf{n} =\mathbf{e}_z$. 
By differentiating in polar coordinates, one can conclude 
that this state shows $\langle S_x\rangle = \langle S_y\rangle =0$. 
Thus we have $\langle S_x\rangle^2 + 
\langle S_y\rangle^2 + \langle S_z\rangle^2 
= \langle S_z\rangle^2$.  
Using $0 \leq \langle S_z\rangle^2 \leq l^2$, 
we conclude that
\begin{eqnarray*}
l \leq (\Delta S_x)^2 + (\Delta S_y)^2 + (\Delta S_z)^2\leq l(l+1).
\end{eqnarray*} 
Thus we obtain the bound
\begin{align*}
&\Vert [\Lambda(\sigma_x), \sigma_x]\Vert 
+\Vert [\Lambda(\sigma_y), \sigma_y]\Vert 
+ \Vert [\Lambda(\sigma_z), \sigma_z]\Vert\\
&\leq 2 \sqrt{l(l+1)} 
(\Vert \id- \Lambda(\sigma_x)^2\Vert
+ \Vert \id- \Lambda( \sigma_y)^2\Vert 
+\Vert \id - \Lambda(\sigma_z)^2\Vert)^{1/2}, 
\end{align*}
where we used the Cauchy-Schwarz inequality. 
One can confirm, 
as expected, that any realizable non-covariant channel 
is inevitably dissipative, 
as non-dissipative (=unitary) dynamics 
satisfies $\id = \Lambda(\sigma_x)^2 
=\Lambda(\sigma_y)^2
= \Lambda(\sigma_z)^2$.
The right-hand side can be regarded as 
a quantity measuring the ``dissipativity" of $\Lambda$, 
while the left-hand side represents the ``magnitude'' of 
dynamics. 
If the environment consists of $N$ qubits, as $l =\frac{N}{2}$ holds the 
term $\sqrt{l(l+1)}$ in the 
right-hand side of the above inequality is proportional to $N$. 
Thus for $\Lambda$ whose magnitude of dynamics 
is $O(1)$, its dissipativity cannot be smaller than 
$O\left(\frac{1}{N}\right)$ in the presence of 
$N$ environment qubits. 
\par
We employ Stokes parameterization 
\cite{HeinosaariZiman}
to illustrate possible channels.
Any qubit state is written as 
$\rho = \frac{1}{2}(\id_{\Sy} + \mathbf{x}\cdot \mathbf{\sigma})$ 
with $|\mathbf{x}|\leq 1$. 
$\Lambda^*$, the dual of $\Lambda$, maps 
$\rho$ to another state $\rho' = \frac{1}{2}
(\id_{\Sy} + \mathbf{y} \cdot \mathbf{\sigma})$. 
This map $(1,\mathbf{x}) \mapsto (1,\mathbf{y})$ is a linear map
on $\mathbf{R}^4$ since
$\Lambda$ is self-adjoint. We denote this map by $\ ^t\tilde{T}_{\Lambda}$ with 
a parameterization,
\begin{eqnarray*}
\ ^t\tilde{T}_{\Lambda}= \left(
\begin{array}{cccc}
1&0&0&0\\
t_1& t_{11}& t_{12}& t_{13}\\
t_2& t_{21}& t_{22}& t_{23}\\
t_3& t_{31}& t_{32}& t_{33}
\end{array}
\right)
=\left(
\begin{array}{cc}
1& \mathbf{0} \\
\mathbf{t} & T
\end{array}
\right),
\end{eqnarray*}
where $T$ is a $3\times 3$ matrix. 
Back in the Heisenberg picture, 
we obtain 
\begin{eqnarray*}
\Lambda( a_0 \id_{\Sy} + \mathbf{a} \cdot \mathbf{\sigma})
= (a_0 + \mathbf{t}\cdot \mathbf{a}) \id_{\Sy} 
+ (T \mathbf{a})\cdot \mathbf{\sigma}. 
\end{eqnarray*}
$T$ can be written as 
\begin{eqnarray*}
T= R_1 D R_2, 
\end{eqnarray*}
where $R_1$ and $R_2$ are elements of $SO(3)$ and 
$D$ is a diagonal matrix as, 
\begin{eqnarray*}
D= \left(
\begin{array}{ccc}
\lambda_1 &0&0\\
0&\lambda_2&0\\
0&0& \lambda_3
\end{array}
\right).
\end{eqnarray*}
One can choose the coordinate system so that 
$R_2= \id$ is satisfied. Thus 
we will consider $T$ with form $T= RD$. 
Then we obtain, for redefined $\mathbf{t}$,  
\begin{eqnarray*}
\Lambda( a_0 \id_{\Sy} + \mathbf{a}\cdot \mathbf{\sigma}) 
= (a_0 + \mathbf{t}\cdot \mathbf{a} )\id_{\Sy}
+ \sum_{ij=1}^3 R_{ij} \lambda_j a_j \sigma_i, 
\end{eqnarray*}
where $R_{ij} \in SO(3)$. 
Assume that $R$ is written as a rotation around the $z$-axis, with the vector ${\bf t}= {\bf 0}$, as,
\begin{eqnarray*}
R=\left(
\begin{array}{ccc}
\cos \theta & \sin \theta &0\\
-\sin \theta & \cos \theta &0 \\
0&0& 1 \\
\end{array}
\right).
\end{eqnarray*}
Then we have 
\begin{eqnarray*}
\lambda_x |\sin \theta| \leq (\Delta_{\rho_{\R}} S_x) \sqrt{1 -\lambda_x^2}\\
\lambda_y |\sin \theta| \leq (\Delta_{\rho_{\R}} S_y) \sqrt{1 -\lambda_y^2}.
\end{eqnarray*}
That is, we have a relation between the dissipative and 
symmetry breaking natures.  
\begin{eqnarray*}
\lambda_x^2 \leq \frac{(\Delta_{\rho_{\R}} S_x)^2}
{(\Delta_{\rho_{\R}} S_x)^2 +(\sin \theta)^2};\\
\lambda_y^2 \leq \frac{(\Delta_{\rho_{\R}} S_y)^2}
{(\Delta_{\rho_{\R}} S_y)^2 +(\sin \theta)^2}. 
\end{eqnarray*}

\section{Concluding remarks}
We have seen that there is a positive lower bound on the difference between an arbitrary quantum channel and the restriction of a covariant channel, and moreover, that in order to reduce this discrepancy a large spread in the generator of the symmetry is needed in the reference system. This result bears similarities with the WAY theorem, and is in line with the relational view of quantum mechanics, wherein we interpret non-symmetric channels as representatives of their symmetric counterparts of system and reference taken together. The large spread required for good approximation of relative (symmetric) by non-relative (asymmetric) can be understood as a condition on the quality of the reference frame, in the sense of the findings of \cite{lmb} and \cite{mlb}. 
As a final remark, we mention that 
there is yet another symmetry condition 
on channels that differs from the one employed 
in this paper and arises naturally in the context of quantum reference frames. 
We will return to this issue elsewhere. 
\section*{Acknowledgments}
TM acknowledges financial support from JSPS (KAKENHI Grant Number 20K03732).

\end{document}